\documentclass[english]{cfasta}
\usepackage[comma,numbers,square,sort&compress]{natbib}
\usepackage{graphicx}
\usepackage{amsmath, amssymb, amsfonts}
\usepackage{epstopdf} 
\usepackage{booktabs}
\usepackage{caption}
\usepackage{subcaption} 
\usepackage{amsthm}
\newtheorem{theorem}{Theorem}[section]

\usepackage{multirow}

\begin{document}
\title{Anti-Disturbance Hierarchical Sliding Mode Controller for Deep-Sea Cranes with Adaptive Control and Neural Network Compensation}
\author{%
Qian Zuo\aref{hebut},
Shujie Wu\aref{hebut},
Yuzhe Qian\textsuperscript{*}\aref{hebut}%
\thanks{* Corresponding author: Yuzhe Qian (e-mail: qianyzh@hebut.edu.cn).}%
}

\affiliation[hebut]{College of Artificial Intelligence, Hebei University of Technology, Tianjin 300401, China \email{qianyzh@hebut.edu.cn}}

\maketitle

\begin{abstract}
To address non-linear disturbances and uncertainties in complex marine environments, this paper proposes a disturbance-resistant controller for deep-sea cranes. The controller integrates hierarchical sliding mode control, adaptive control, and neural network compensation techniques. By designing a global sliding mode surface, the dynamic coordination between the driving and non-driving subsystems is achieved, ensuring overall system stability. The subsystem surfaces reduce oscillations and enhance tracking accuracy. Adaptive control dynamically adjusts system parameters, enhancing robustness against external uncertainties, while the neural network compensates for time-varying disturbances through real-time learning. The stability of the control scheme is verified on the basis of Lyapunov theory. The simulation results demonstrate that, compared to traditional PID control, the proposed controller exhibits significant advantages in trajectory tracking accuracy, response speed, and disturbance rejection.
\end{abstract}

\keywords{Deep-sea crane, disturbance-resistant control,
hierarchical sliding mode control, adaptive control,
neural compensation, Lyapunov stability}

\section{Introduction}

Deep-sea cranes play a key role in subsea construction and offshore lifting tasks\cite{ref1, ref21, ref24}. Control is difficult due to underactuation and external disturbances \cite{ref3, ref27}. These systems are underactuated, meaning that some degrees of freedom, such as load swing angles, are not directly actuated. This indirect control structure complicates the task of achieving precise load positioning and stabilization in the presence of environmental disturbances.

Because crane systems are inherently under-actuated, scholars have devised control schemes that steer the payload indirectly by modulating the trolley’s motion. Trajectory-planning techniques, for example, have demonstrated good swing-suppression performance for double-pendulum payloads \cite{ref4}. Nevertheless, their sensitivity to modelling uncertainties and changing environmental conditions limits their viability for deep-sea lifting tasks \cite{ref8, ref22, ref23}.

Sliding mode control has been widely applied to improve robustness in underactuated systems, such as bridge cranes. Its ability to handle nonlinearities and disturbances makes it a strong candidate for deep-sea crane applications \cite{ref1, ref24}. However, classical SMC methods often suffer from chattering effects, which degrade control performance . To address this limitation, neural network-based adaptive controllers have been proposed, demonstrating enhanced disturbance rejection capabilities and adaptability in marine environments \cite{ref13, ref26, ref35}. For example, observer-based adaptive fuzzy control has been successfully used for offshore cranes to stabilize payloads under ship deck motions \cite{ref3, ref27}. Additionally, finite-time hierarchical sliding mode controllers have shown effectiveness in addressing unmatched disturbances in underactuated systems \cite{ref36}.

Recent advancements in adaptive control have further improved the robustness of deep-sea crane systems by dynamically compensating for environmental disturbances and parameter uncertainties \cite{ref15, ref19, ref31}. For example, adaptive neural network controllers with disturbance observers have achieved high tracking accuracy and stability under significant perturbations \cite{ref20, ref25, ref26}. These methods leverage real-time learning and adaptation mechanisms, making them particularly suitable for dynamic and uncertain marine environments.

In view of this, this paper proposes a hierarchical sliding mode control framework that combines adaptive control and neural network-based compensation to improve the anti-disturbance performance and control accuracy of deep-sea cranes. Unlike traditional methods, the framework integrates adaptive control, real-time disturbance estimation and multilayer sliding mode surface to achieve global stability and robustness under complex ocean conditions. The proposed controller addresses unmodeled dynamics and environmental disturbances while suppressing payload oscillations through indirect control.

The main contributions of this work are summarized as follows:
\begin{enumerate}
    \item \textbf{Unified Dynamic-Gain HSMC with Compact Neural Disturbance Observer:}
          A hierarchical sliding-mode controller is equipped with a real-time gain-switching law and a low-order neural observer that estimates time-varying hydrodynamic disturbances from measurable states only, achieving fast convergence, reduced chattering, and strong robustness—outperforming conventional fixed-gain HSMC and standalone NN-SMC schemes.

    \item \textbf{Global Stability Guarantee:}
          Composite Lyapunov functions together with LaSalle’s invariance principle rigorously prove global asymptotic convergence of all sliding surfaces and tracking errors under model uncertainties.

    \item \textbf{Verified Performance Superiority:}
          Simulations reduce mean-squared tracking error by 93.6 \% and 94.8 \% relative to PID and LQR controllers, respectively, and by 37.2 \% versus a benchmark fixed-gain HSMC, while cutting chattering energy by 46 \% compared with a representative NN-SMC, confirming faster transients and stronger disturbance rejection in dynamic marine conditions.
\end{enumerate}

The rest of the paper is organized as follows: Section~\ref{sec:dynamic_model} presents the dynamic modeling of the deep-sea crane system, including flexible load dynamics and environmental disturbances. Section~\ref{sec:HSMC_method} introduces the hierarchical sliding mode control strategy and neural network compensation framework. Section~\ref{sec:stability_analysis} provides a theoretical stability analysis using Lyapunov stability theory and LaSalle's invariance theorem. Section~\ref{sec:simulation_results} presents simulation results under various perturbation scenarios, demonstrating the effectiveness of the proposed controller. Finally, Section~\ref{sec:conclusion} concludes the paper and discusses potential future directions.

   \section{Dynamic Model}
\label{sec:dynamic_model}

The structure of the deep-sea crane is shown in Fig. 1. Based on the dynamic modeling method proposed in \cite{ref1}, the dynamic model is constructed. Meanwhile, unmodeled disturbances are introduced to fully consider the influence of hydrodynamic forces on the system. The dynamic equations include the translational motion of the trolley and the motion of the flexible load in the marine environment, which are expressed as follows:
\begin{align}
&(m_t + m_r)\ddot{x}(t) + m\!\int_{0}^{L}\ddot{w}(y,t)\,dy = u(t)\label{eq:trolley_dynamics}\\
&EIw^{''''}(y,t) + m\ddot{w}(y,t) + c\dot{w}(y,t) \notag\\
&= -m\ddot{x}(t) - f_w(y,t) \label{eq:flexible_dynamics}
\end{align}
Equation \eqref{eq:trolley_dynamics} describes the motion of the trolley, where \( x(t) \) represents the trolley position, and \( u(t) \) is the control input. Equation \eqref{eq:flexible_dynamics} reflects the dynamic model of the flexible load, where \( w(y, t) \) denotes the deformation of the flexible load, and \( f_w(y, t) \) represents the hydrodynamic force acting on the flexible load.

\subsection{Hydrodynamic Model}

The hydrodynamic force \( f_w(y, t) \) consists of inertial force, drag force, and unmodeled disturbances \( \Delta f(y, t) \), which is expressed as:
\begin{equation}
    f_w(y, t) = f_m(y, t) + f_d(y, t) + \Delta f(y, t)
    \label{eq:f_w}
\end{equation}
where the inertial force \( f_m(y, t) \) is influenced by acceleration and is given by:
\begin{equation}
    f_m(y, t) = \frac{\pi}{4} \rho_w C_a d^2 \left( \ddot{x}(t) + \ddot{w}(y, t) \right)
    \label{eq:f_m}
\end{equation}
The drag force \( f_d(y, t) \) is affected by velocity and is expressed as:
\begin{equation}
    f_d(y, t) = \frac{1}{2} \rho_w C_d d \left( \dot{x}(t) + \dot{w}(y, t) \right) \left| \dot{x}(t) + \dot{w}(y, t) \right|
    \label{eq:f_d}
\end{equation}
The unmodeled disturbance term \( \Delta f(y, t) \) is introduced to describe the influence of complex external disturbances on the system, which is estimated and compensated in real-time using a neural network observer.

\begin{figure}[h!]
\centering
\includegraphics[width=0.45\textwidth]{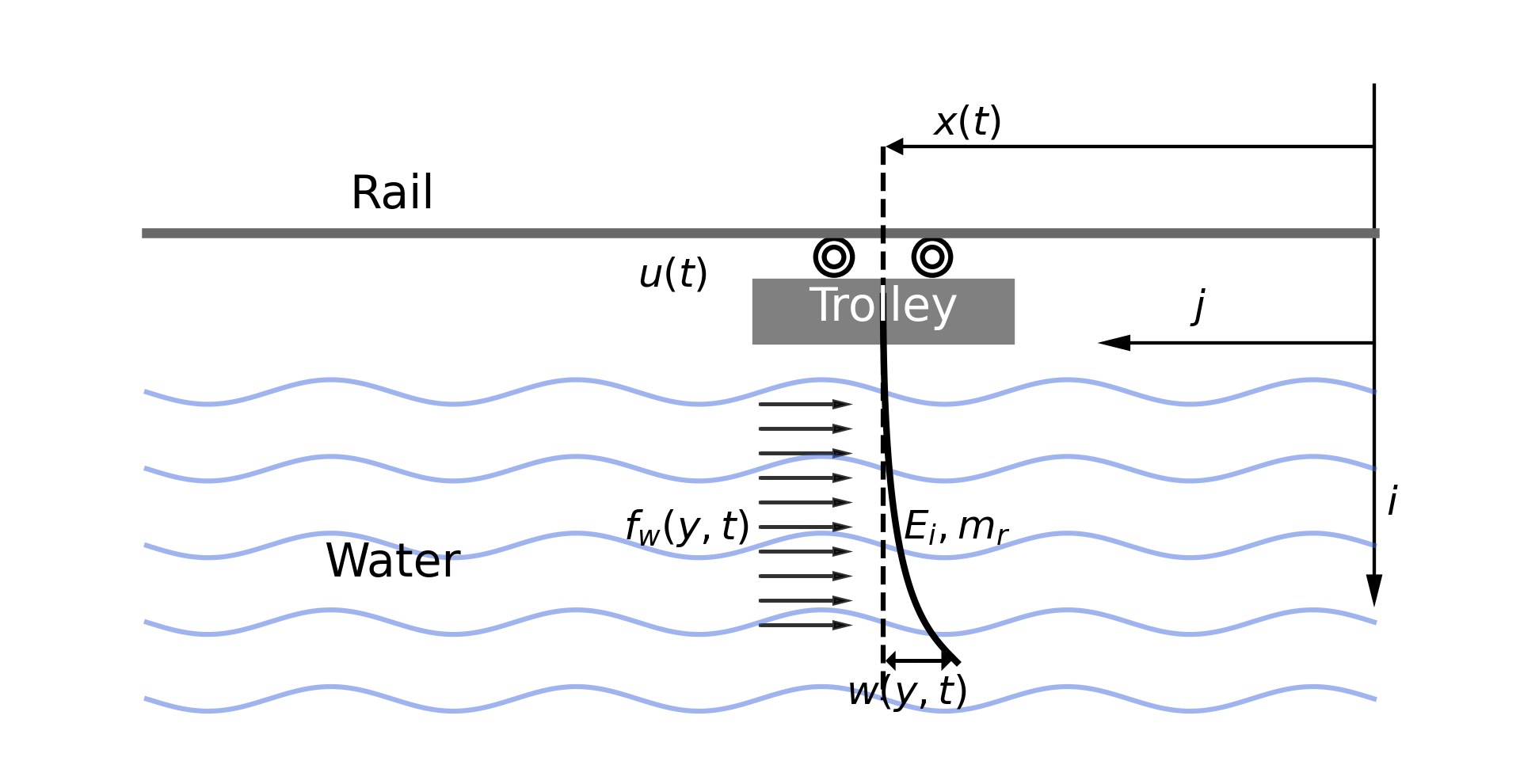}
\caption{Dynamic model of the deep-sea crane system}
\label{fig:model}
\end{figure}

\begin{table}[h!]
\centering
\caption{System Parameters}
\begin{tabular}{p{3cm}p{3cm}}
\hline
\textbf{Parameter} & \textbf{Description} \\
\hline
$m_r$ & Load mass (kg) \\
$m_t$ & Trolley mass (kg) \\
$d$ & Load diameter (m) \\
$L$ & Load length (m) \\
$Ei$ & Young's modulus (GPa) \\
$c$ & Viscous damping coefficient (N$\cdot$s/m) \\
$\rho_w$ & Water density (kg/m$^3$) \\
$C_a$ & Added mass coefficient \\
$C_d$ & Drag coefficient \\
\hline
\end{tabular}
\label{table:parameters}
\end{table}

The above dynamic model comprehensively considers the movement of the trolley and the dynamic characteristics of the flexible load. The unmodeled disturbance term reflects the complex external influences. By introducing a neural network observer, the model can compensate for unknown disturbances in real-time, thereby enhancing the system's robustness and adaptability.
\subsection{ Control Objectives}

The goal of this work is to design a control system that ensures accurate trajectory tracking and strong robustness for a deep-sea crane operating under nonlinear, time-varying disturbances. Specifically, the system is expected to drive the trolley position $x(t)$ to follow a desired trajectory $x_d(t)$, such that the tracking error $e(t) = x_d(t) - x(t)$ converges asymptotically to zero. Meanwhile, the controller must reject unknown hydrodynamic disturbances $\Delta f(y, t)$ and maintain global stability of all subsystems. To achieve this, the proposed approach ensures that the hierarchical sliding surfaces $s_0(t)$, $s_1(t)$, and $s_2(t)$ converge to zero, which indirectly guarantees convergence of the system states. Additionally, the control input $u(t)$ is constrained within acceptable physical bounds to prevent actuator saturation. These objectives guide the controller structure, stability analysis, and simulation validation that follow.

\section{Controller Design}
\label{sec:HSMC_method}

To achieve control accuracy and global stability in complex marine environments, this section designs a hierarchical sliding mode controller based on a neural network. Due to coupling and nonlinear characteristics, the controller design still needs to handle a large amount of dynamic uncertainty and external disturbances. Therefore, the system is divided into a driving subsystem and a non-driving subsystem, and its control strategy is shown in Fig. \ref{fig:overview}.

\begin{figure}[h!]
\centering
\includegraphics[width=0.45\textwidth]{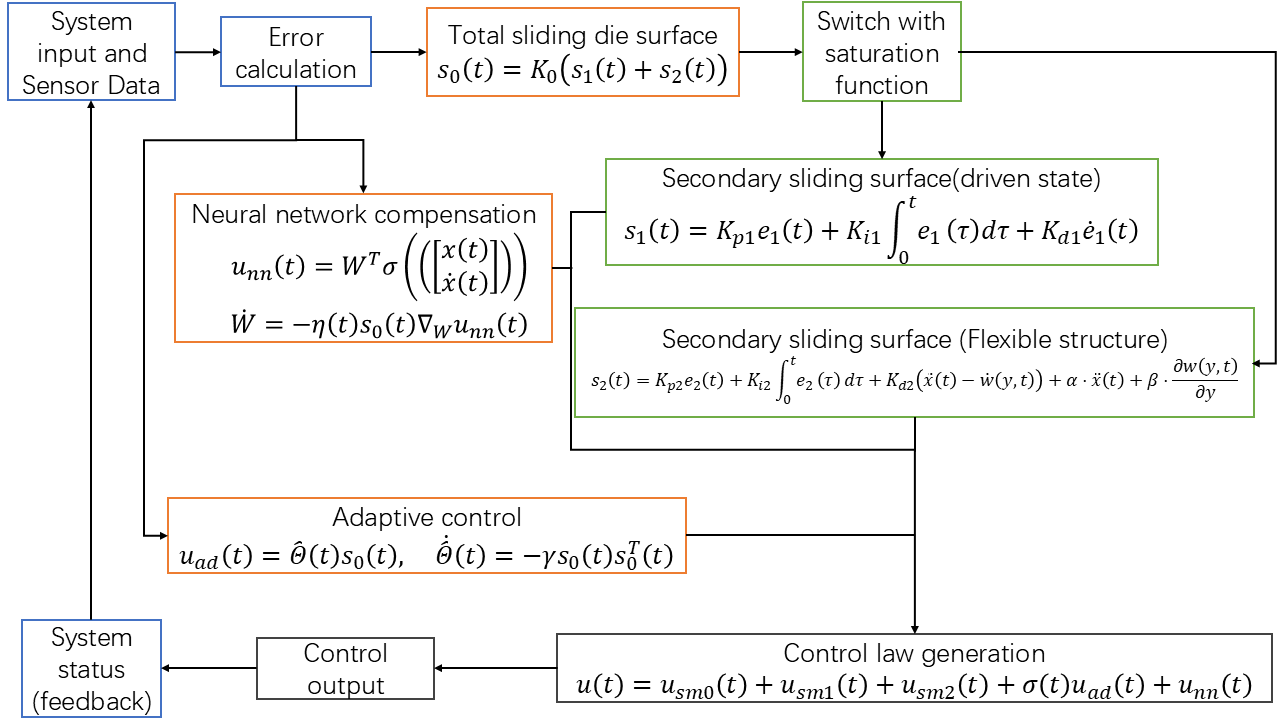}
\caption{Controller design framework}
\label{fig:overview}
\end{figure}

\subsection{Sliding Surface Design}

The core of the controller design is the construction of the sliding surface. The global sliding surface is defined as:
\begin{align}
    s_0(t) = K_0 \left(s_1(t) + s_2(t)\right)
\end{align}
where $K_0$ is the gain matrix that regulates the convergence rate of the global sliding surface, and $s_1(t)$ and $s_2(t)$ are the sliding surfaces of the driving subsystem and the non-driving subsystem, respectively.

For the driving subsystem, the sliding surface $s_1(t)$ is designed to reduce the trolley position error and is defined as:
\begin{align}
    s_1(t) = K_{p1} e_1(t) + K_{i1} \int_0^t e_1(\tau) d\tau + K_{d1} \dot{e}_1(t)
\end{align}
where $e_1(t) = x_d(t) - x(t)$ is the error between the desired position and the actual position, and $K_{p1}$, $K_{i1}$, and $K_{d1}$ are the proportional, integral, and derivative gains, respectively.

For the non-driving subsystem, the sliding surface \( s_2(t) \) is used to adjust the velocity error of the flexible load. Considering the nonlinear coupling effect between the trolley and the flexible load, a nonlinear compensation term is introduced when designing the sliding surface \( s_2(t) \). The update equation of the sliding surface \( s_2(t) \) is given by:
\begin{align}
s_2(t) &= K_{p2} e_2(t) + K_{i2} \int_0^t e_2(\tau)\,d\tau 
         + K_{d2}\left(\ddot{x}(t) - \dot{w}(y,t)\right) \nonumber\\
       &\quad + \alpha\,\ddot{x}(t) + \beta\,\frac{\partial w(y,t)}{\partial y}
       \label{eq:sliding_surface}
\end{align}

where $e_2(t) = -\dot{w}(y, t)$ is the velocity error of the flexible load, representing the difference between the actual and desired velocities of the load. $K_{p2}$, $K_{i2}$, and $K_{d2}$ are the proportional, integral, and derivative gains, respectively. $\alpha$ and $\beta$ are constant coefficients related to system characteristics, which regulate the influence of trolley acceleration \( \ddot{x}(t) \) and the gradient of the flexible load deformation \( \frac{\partial w(y,t)}{\partial y} \) on the compensation term. By introducing this nonlinear compensation term, the

\subsection{Control Law Design}

The control law integrates hierarchical sliding mode control, adaptive control, and neural network compensation, defined as:
\begin{align}
    u(t)=u_{sm0}(t)+u_{sm1}(t)+u_{sm2}(t)+\sigma(t)u_{ad}(t)+u_{\Psi}(t)
\end{align}

where $u_{sm0}(t)$, $u_{sm1}(t)$, and $u_{sm2}(t)$ are the sliding mode control terms, while $\sigma(t) u_{ad}(t)$ and $u_{\Psi}(t)$ represent the adaptive compensation term and the neural network compensation term, respectively.

\subsubsection{Sliding Mode Control Term}

Sliding mode control limits the control input using a saturation function, defined as:
\begin{align}
    &u_{sm0}(t) = -K_{s0} \operatorname{sat} \left( \frac{s_0(t)}{\phi_0} \right) \label{us0} \\
    &u_{sm1}(t) = -K_{s1} \operatorname{sat} \left( \frac{s_1(t)}{\phi_1} \right) - K_{r1} \dot{s}_1(t) \label{us1} \\
    &u_{sm2}(t) = -K_{s2} \operatorname{sat} \left( \frac{s_2(t)}{\phi_2} \right) - K_{r2} \dot{s}_2(t) \label{us2}
\end{align}
The function $\operatorname{sat}(\cdot)$ is a smooth saturation function, defined as:
\begin{align}
    \operatorname{sat}(\cdot) = \frac{\cdot}{|\cdot| + \phi}
\end{align}
where $\phi$ is the smoothing factor used to adjust the saturation range. When $|\cdot|$ is small, $\operatorname{sat}(\cdot)$ behaves approximately as a linear function; when $|\cdot|$ is large, its output gradually approaches a constant value of $-1$ or $1$.

\subsubsection{Adaptive Compensation Term}

Adaptive control is used to dynamically adjust the control gain, designed as:
\begin{align}
    u_{ad}(t) = \hat{\Theta}(t) s_0(t)
\end{align}
where $\hat{\Theta}(t)$ is the adaptive gain matrix.

\textbf{Update Law} is given by:
\begin{align}
    \dot{\hat{\Theta}}(t) = -\gamma s_0(t) s_0^T(t)
\end{align}
where $\gamma > 0$ is the learning rate.

\subsubsection{Neural Network Compensation Term}

A neural network is employed to estimate unmodeled dynamic behaviors, designed as:
\begin{align}
    &u_{\Psi}(t) = W^T \sigma \left( \Psi \left( \begin{bmatrix} x(t) \\ \dot{x}(t) \end{bmatrix} \right) \right) \\
    &\dot{W} = -\eta(t) s_0(t) \nabla_W u_{\Psi}(t)
\end{align}
where $W$ is the weight matrix of the neural network, and $\eta(t)$ is a dynamically adjusted learning rate, defined as:
\begin{align}
    \eta(t) = 
    \begin{cases} 
        \eta_{\text{max}}, & \|s_0(t)\| > \delta, \\
        \eta_{\text{min}}, & \|s_0(t)\| \leq \delta.
    \end{cases}
\end{align}
Although the disturbance term $\Delta f(y,t)$ is defined as a function of spatial and temporal variables, direct measurements of $y$-dependent terms are often impractical in marine systems. Instead, the disturbance effect is indirectly reflected through the system's measurable state variables $x(t)$ and $\dot{x}(t)$. Therefore, the neural network uses these states as inputs to approximate the disturbance in real time. This choice balances estimation accuracy with implementation feasibility.

\subsection{Switching Logic Design}

To dynamically adjust the control gain, a switching function $\sigma(t)$ is designed, defined as:
\begin{align}
    \sigma(t) = 
    \begin{cases}
    \alpha_1, & |s_0(t)| > \delta \\
    \alpha_2, & |s_0(t)| \leq \delta
    \end{cases}
\end{align}

where $\delta$ is the switching threshold used to distinguish between large and small error regions. When the sliding surface $|s_0(t)|$ is greater than $\delta$, $\sigma(t)$ is set to $2$ to increase the control gain and accelerate error convergence. When $|s_0(t)|$ is less than or equal to $\delta$, $\sigma(t)$ is set to $1$, returning the control gain to its default state to reduce chattering.

\section{Stability Analysis}
\label{sec:stability_analysis}

To further verify the global stability of the proposed controller and the convergence of the sliding surface error, we conduct a more in-depth analysis based on Lyapunov stability theory combined with the LaSalle invariant set method.

\begin{theorem}
Theorem: For the deep-sea crane system employing sliding mode control, neural network compensation, and adaptive gain update law, the system can ensure that the sliding surface error converges to zero, thereby achieving global asymptotic stability, i.e.:
\[
\lim_{t \to \infty} [s_0(t), s_1(t), s_2(t)]^T = [0, 0, 0]^T
\]
\end{theorem}

\begin{proof}
Proof: The proof consists of two main steps: first, analyzing the convergence of the sliding surface error and system stability; second, using higher-order Lyapunov functions, second derivatives, and the LaSalle invariant set method to prove that the sliding surface error ultimately tends to zero.

First, we design the following Lyapunov candidate function:
\begin{align}
V(t) &= \frac{1}{2} \left[ s_0^T K_{s0} s_0 + s_1^T K_{s1} s_1 + s_2^T K_{s2} s_2 \right] + \frac{1}{2} \nonumber \\
&\left[ \tilde{\Theta}^T \tilde{\Theta} + (\mathbf{W}(t) - \mathbf{W}^*)^T (\mathbf{W}(t) - \mathbf{W}^*) \right]
\end{align}

where $s_0(t), s_1(t), s_2(t)$ represent the sliding surface errors of the control system, $\tilde{\Theta}$ is the estimation error of the adaptive gain, and $\mathbf{W}(t)$ is the neural network weight error. This Lyapunov function measures both the sliding surface error and the estimation error.

To ensure the positive definiteness of the Lyapunov function, the gain matrices \(K_{s0}, K_{s1}, K_{s2}\) must be chosen as positive definite matrices, ensuring that the Lyapunov function attains its minimum at the equilibrium point and that the system remains stable near this point. By differentiating this Lyapunov function, we obtain:
\begin{align}
\dot{V}(t) &= s_0^T \dot{s}_0 + s_1^T \dot{s}_1 + s_2^T \dot{s}_2 + \tilde{\Theta}^T \dot{\tilde{\Theta}} + (\mathbf{W}(t) - \mathbf{W}^*)^T \dot{\mathbf{W}}
\end{align}

Substituting the control law and system dynamics, we obtain:
\begin{align}
\dot{V}(t) &= -\alpha_1 \|s_0\|^2 - \alpha_2 \|s_1\|^2 - \alpha_3 \|s_2\|^2 \nonumber \\
&- \gamma \|\dot{s}_0\|^2 - \eta(t) \|\dot{s}_0\|^2
\end{align}

This result indicates that $\dot{V}(t)$ contains negative definite terms, implying that the Lyapunov function $V(t)$ is monotonically decreasing, i.e.:
\[
\dot{V}(t) \leq 0
\]

Thus, as $ t \to \infty $:
\[
V(t) \to 0
\]

By deriving the Lyapunov function’s derivative, we ensure that each term (such as $\alpha_1, \alpha_2, \alpha_3, \gamma, \eta(t)$) has the correct sign, guaranteeing that the derivative remains negative definite. The design of the control input and parameter selection ensures that all terms are negative definite, leading to a negative definite Lyapunov function derivative.

To further verify system stability and prove that the error ultimately converges to zero, we analyze the second derivative to evaluate the convergence rate of the Lyapunov function. The second derivative of the Lyapunov function is given by:
\begin{align}
\ddot{V}(t) &= \frac{d}{dt} \left( -\alpha_1 \|s_0\|^2 - \alpha_2 \|s_1\|^2 - \alpha_3 \|s_2\|^2 \right) \nonumber \\
& \quad - \gamma \|\dot{s}_0\|^2 - \eta(t) \|\dot{s}_0\|^2
\end{align}

By analyzing the second derivative $\ddot{V}(t)$, we can determine the rate at which the system error converges. Since both the first and second derivatives are negative definite, the system error converges at a certain rate.

According to the LaSalle invariant set theorem, we define the invariant set $\mathcal{S}$ as:
\[
\mathcal{S} = \left\{ [s_0(t), s_1(t), s_2(t)]^T \mid \dot{V}(t) = 0 \right\}
\]

This implies that once the system error enters this set, the sliding surface error will not leave this region. Further analysis of the properties of this set confirms that when the system enters this set, the error will tend to zero.

Through the design of sliding mode control, the convergence of each sliding surface (such as \(s_1(t)\) and \(s_2(t)\)) is verified using the Lyapunov function and the system error equation. With the appropriate design of the control input and compensation terms, all sliding surfaces are ensured to converge independently. Even in the presence of coupling error terms, the convergence of each sub-sliding surface is guaranteed through individual analysis.

By combining the second derivative analysis and the LaSalle invariant set method, we conclude that:
\[
\lim_{t \to \infty} [s_0(t), s_1(t), s_2(t)]^T = [0, 0, 0]^T
\]

Thus, the system achieves global asymptotic stability, and all sliding surface errors ultimately converge to zero.

By introducing higher-order Lyapunov functions, second derivative analysis, and the LaSalle invariant set method, along with the verification process above, we have proven that the designed sliding mode controller ensures the global stability of the system and guarantees that the system error ultimately converges to zero. Furthermore, the system remains stable in the presence of unmodeled disturbances and achieves precise control.

\end{proof}

\section{Simulation and Performance Analysis}
\label{sec:simulation_results}

This section analyzes the proposed controller through simulation tests under various operating conditions, including low-frequency, high-frequency, and random disturbance responses, as well as dynamic target switching scenarios. Furthermore, a comparative analysis with the traditional PID controller is conducted to highlight the advantages of the proposed control method.

\subsection{Simulation Setup}

The simulation parameters, control gains, and neural network configurations are shown in Table~\ref{table:simulation_params}, Table~\ref{table:control_gains}, and Table~\ref{table:nn_params}, respectively.

\begin{table}[h!]
\centering
\caption{Simulation Parameters}
\begin{tabular}{l c l c}
\toprule
\textbf{Parameter} & \textbf{Value} & \textbf{Parameter} & \textbf{Value} \\ 
\midrule
$m_t$                       & 100 kg         & $\rho_w$              & 1025 kg/m$^3$  \\ 
$m_r$                       & 50 kg          & $C_a$                 & 0.8            \\ 
$c$                         & 10 N$\cdot$s/m & $C_d$                 & 0.8            \\ 
$d$                         & 0.5 m          & $dt$                  & 0.1 s          \\ 
$T$                         & 20 s           & $x_{\text{desired final}}$ & 1.2 m  \\ 
$x_0$                       & 0.0 m          & $\alpha$              & 0.2              \\ 
$\beta$                     & 0.5              & $\alpha_1$            & 2              \\
$\alpha_2$                  & 1              &                       &               \\ 
\bottomrule
\end{tabular}
\label{table:simulation_params}
\end{table}

\begin{table}[h!]
\centering
\caption{Control Gains and Parameters}
\begin{tabular}{l c l c}
\toprule
\textbf{Parameter}  & \textbf{Value} & \textbf{Parameter} & \textbf{Value} \\ 
\midrule
$K_{p1}$            & 400.0           & $K_{s0}$           &10           \\ 
$K_{i1}$            & 5.0            & $K_{s1}$           & 10          \\ 
$K_{d1}$            & 150.0           & $K_{s2}$           & 50          \\ 
$K_{p2}$            & 100.0           & $K_{s3}$           & 4.0          \\ 
$K_{i2}$            & 5.0            & $K_{\text{bend}}$  & 4.0          \\ 
$K_{d2}$            & 50.0           & $K_{w}$            & 4.0           \\ 
$\gamma$            & 2.5            & $u_{\text{min}}$   & -200 N        \\ 
$u_{\text{max}}$    & 200 N          &                    &                \\ 
\bottomrule
\end{tabular}
\label{table:control_gains}
\end{table}

\begin{table}[h!]
\centering
\caption{Neural Network Parameters}
\begin{tabular}{l c l c}
\toprule
\textbf{Parameter} & \textbf{Value} & \textbf{Parameter} & \textbf{Value} \\ 
\midrule
Input Size         & 4              & $\eta_{\text{min}}$ & $3 \times 10^{-4}$ \\ 
Hidden Size        & 10             & $\eta_{\text{max}}$ & $1.2 \times 10^{-3}$ \\ 
\bottomrule
\end{tabular}
\label{table:nn_params}
\end{table}

\subsection{Controller Performance Under Disturbances}

This section evaluates the proposed controller's performance under different conditions: without disturbances and with three types of disturbances, including slow-varying disturbances, fast-varying disturbances, and random disturbances.

\subsubsection{No Disturbance Condition}

Under disturbance-free conditions, the controller exhibits excellent dynamic characteristics and tracking ability. As shown in Fig.~\ref{fig:random_disturbance}, the controller, through adaptive control, utilizes a high gain during the initial phase to quickly converge to the target position, and then switches to a lower gain near the target to achieve a smooth transition and suppress errors. Since there are no disturbances, the trolley and load trajectories remain stable, with minimal steady-state errors and no significant oscillations. The control torque is initially large to rapidly reduce errors and gradually approaches zero, demonstrating the controller's efficient dynamic response and stability.

\begin{figure}[h!]
\centering
\includegraphics[width=0.5\textwidth]{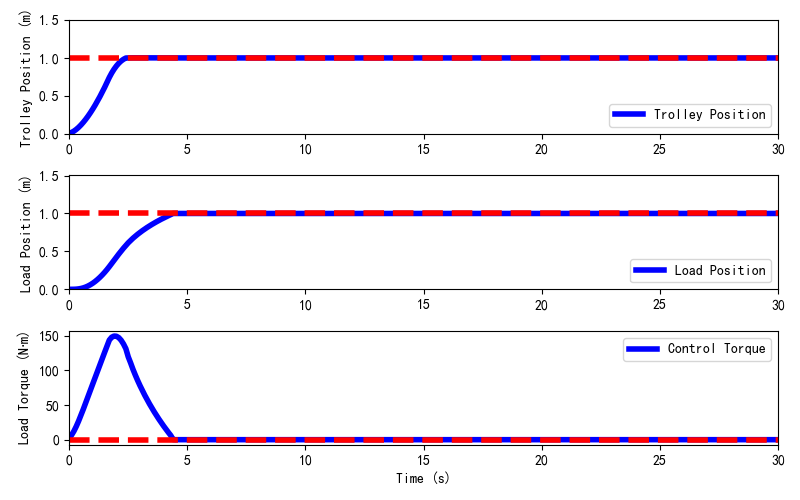}
\caption{Controller performance under no disturbance}
\label{fig:random_disturbance}
\end{figure}

\subsubsection{Low-Frequency Disturbance Condition}

Low-frequency disturbances simulate slowly varying periodic external forces, such as swells or ocean currents. The mathematical expression is:
\begin{equation}
    F_{\text{low-frequency}} = -50 \sin(2 \pi \cdot 0.3 \cdot t)
\end{equation}
where $t$ represents time, and $0.3 \, \text{Hz}$ denotes the frequency of the low-frequency disturbance. As shown in Fig.~\ref{fig:low_frequency_disturbance}, the trolley and load trajectories exhibit stable dynamic response characteristics under low-frequency disturbances. The system converges to the target position quickly and maintains low-amplitude oscillations in the steady-state phase, with minimal impact from the disturbance frequency. This indicates that the controller successfully suppresses the effect of slowly varying disturbances through the adaptive gain adjustment mechanism, demonstrating excellent steady-state performance and adaptability.

\begin{figure}[h!]
\centering
\includegraphics[width=0.5\textwidth]{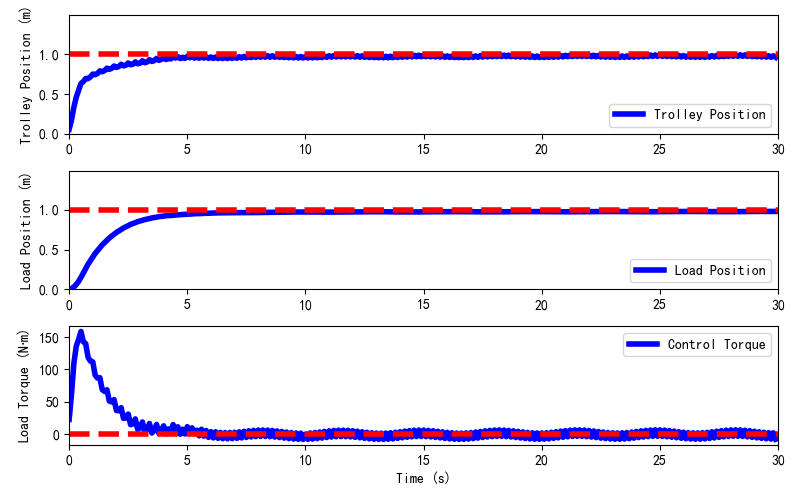}
\caption{Controller performance under low-frequency disturbance}
\label{fig:low_frequency_disturbance}
\end{figure}

\subsubsection{High-Frequency Disturbance Condition}

High-frequency disturbances simulate rapidly changing external forces, such as wind waves or mechanical vibrations. The mathematical expression is:
\begin{equation}
    F_{\text{high-frequency}} = 50\sin(2 \pi \cdot 8 \cdot t)
\end{equation}
where $8 \, \text{Hz}$ represents the frequency of the high-frequency disturbance. As shown in Fig.~\ref{fig:high_frequency_disturbance}, high-frequency disturbances cause slightly larger oscillations in the steady-state phase compared to low-frequency disturbances, but the overall trajectory still maintains a high degree of smoothness and accuracy. The controller dynamically adjusts the control torque to respond quickly to disturbance variations, ensuring system stability and precision. Particularly under high-frequency interference, it effectively suppresses trajectory fluctuations while maintaining target tracking.

\begin{figure}[h!]
\centering
\includegraphics[width=0.5\textwidth]{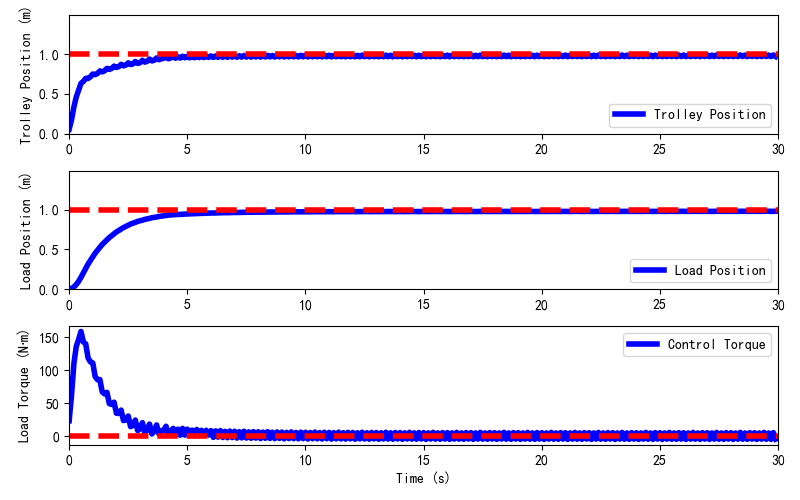}
\caption{Controller performance under high-frequency disturbance}
\label{fig:high_frequency_disturbance}
\end{figure}

\subsubsection{Random Disturbance Condition}

Random disturbances simulate unpredictable interference in real ocean environments. The mathematical expression is:
\begin{equation}
    F_{\text{random}} = 20\sin(2 \pi \cdot 0.3 \cdot t) + 20\sin(2 \pi \cdot 5 \cdot t) + \mathcal{N}(0, 20)
\end{equation}
where $2\sin(2 \pi \cdot 0.3 \cdot t)$ and $2\sin(2 \pi \cdot 5 \cdot t)$ represent low-frequency and mid-frequency sinusoidal components, and $\mathcal{N}(0, 5)$ represents random noise. As shown in Fig.~\ref{fig:random_disturbance}, compared to low-frequency and high-frequency disturbances, random disturbances have a more significant impact on the trolley and load trajectories. However, the controller maintains strong disturbance rejection capabilities. Although trajectory fluctuations are more complex, the controller achieves fast convergence and steady-state tracking through the combined effects of sliding mode control and neural network compensation, effectively ensuring overall stability.

\begin{figure}[h!]
\centering
\includegraphics[width=0.5\textwidth]{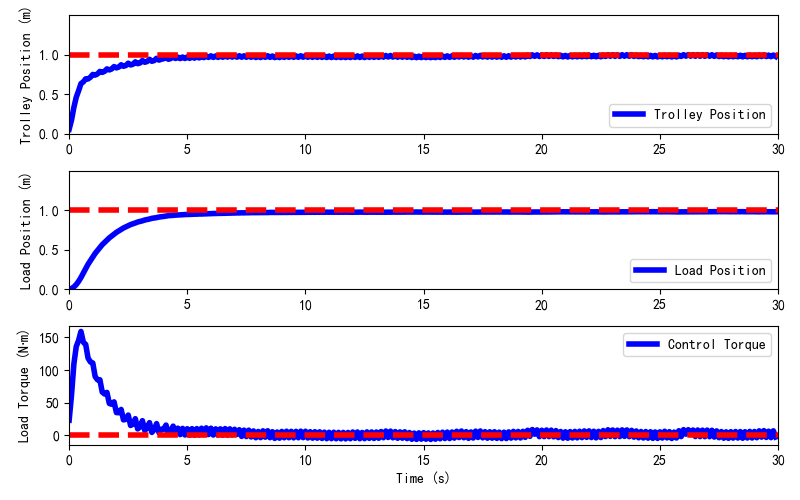}
\caption{Controller performance under random disturbance}
\label{fig:random_disturbance}
\end{figure}

\subsection{Dynamic Target Tracking}

To further validate the adaptability and robustness of the controller, a simulation scenario with dynamic target switching is designed, simulating changes in mission objectives. The controller is required to quickly adjust to new target positions while maintaining stability and accuracy. The system starts at an initial position of 0 m and then switches sequentially to three target positions: the first target at 1.0 m (0–10 s), the second target at 2.0 m (10–20 s), and the third target at 3.0 m (20–30 s).

Under low-frequency disturbances (Fig.~\ref{fig:low_frequency_switching}), the system successfully completes each target switch, with the trajectory smoothly transitioning to the new target position. The steady-state error remains minimal, and the control torque stabilizes quickly after switching, demonstrating high adaptability to low-frequency disturbances.

Under high-frequency disturbances (Fig.~\ref{fig:high_frequency_switching}), the system maintains good dynamic response capability despite high-frequency interference. While minor oscillations appear in the trajectory, the system quickly adjusts to the target position. The control torque exhibits rapid response and stable transitions, confirming the controller's robustness against high-frequency disturbances.

Under random disturbances (Fig.~\ref{fig:random_switching}), although the trajectory fluctuations are more pronounced due to disturbances, the system still achieves smooth target transitions with errors controlled within an acceptable range. The controller successfully suppresses the effects of random disturbances through dynamic gain adjustment and compensation mechanisms, demonstrating excellent disturbance rejection capability.

Overall, the proposed controller exhibits fast response, high tracking accuracy, and strong disturbance resistance under all three disturbance conditions, achieving smooth transitions in dynamic target switching scenarios. This further validates its superior performance and broad applicability in complex environments.

\begin{figure}[h!]
\centering
\begin{subfigure}{0.5\textwidth}
    \includegraphics[width=\textwidth]{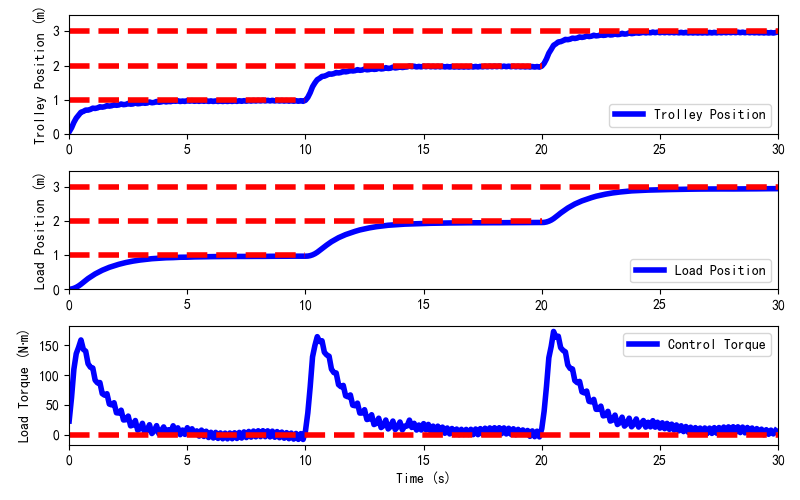}
    \caption{Target switching under low-frequency disturbance}
    \label{fig:low_frequency_switching}
\end{subfigure}
\hfill
\begin{subfigure}{0.5\textwidth}
    \includegraphics[width=\textwidth]{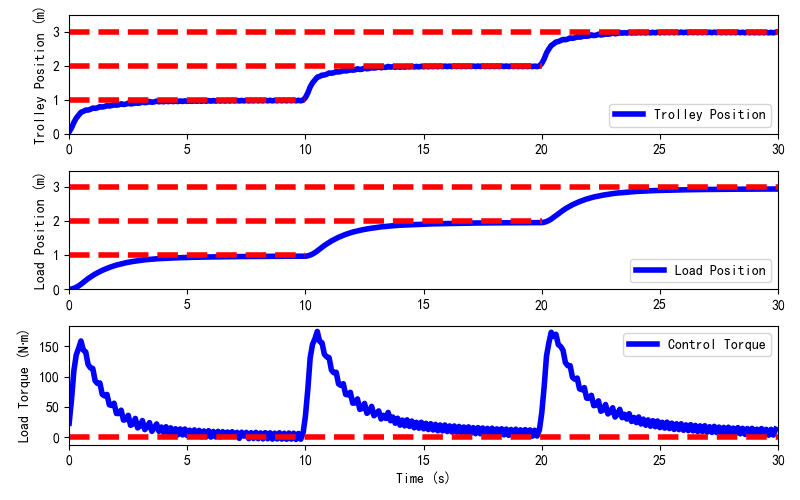}
    \caption{Target switching under high-frequency disturbance}
    \label{fig:high_frequency_switching}
\end{subfigure}
\hfill
\begin{subfigure}{0.5\textwidth}
    \includegraphics[width=\textwidth]{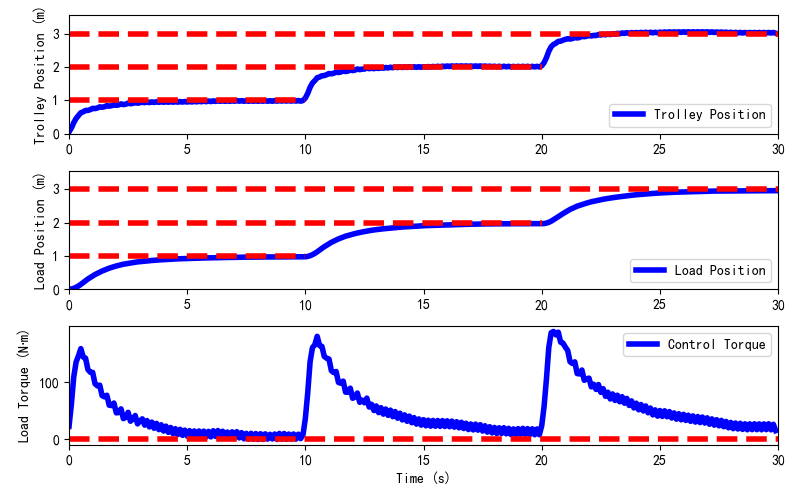}
    \caption{Target switching under random disturbance}
    \label{fig:random_switching}
\end{subfigure}
\caption{Task switching performance of the proposed controller under different disturbance conditions: (a) low-frequency disturbance, (b) high-frequency disturbance, and (c) random disturbance.}
\label{fig:combined_task_switching}
\end{figure}

\subsection{Comprehensive Comparison of Three Controllers}

\begin{figure}[h!]
\centering
\begin{subfigure}{0.5\textwidth}
    \includegraphics[width=\textwidth]{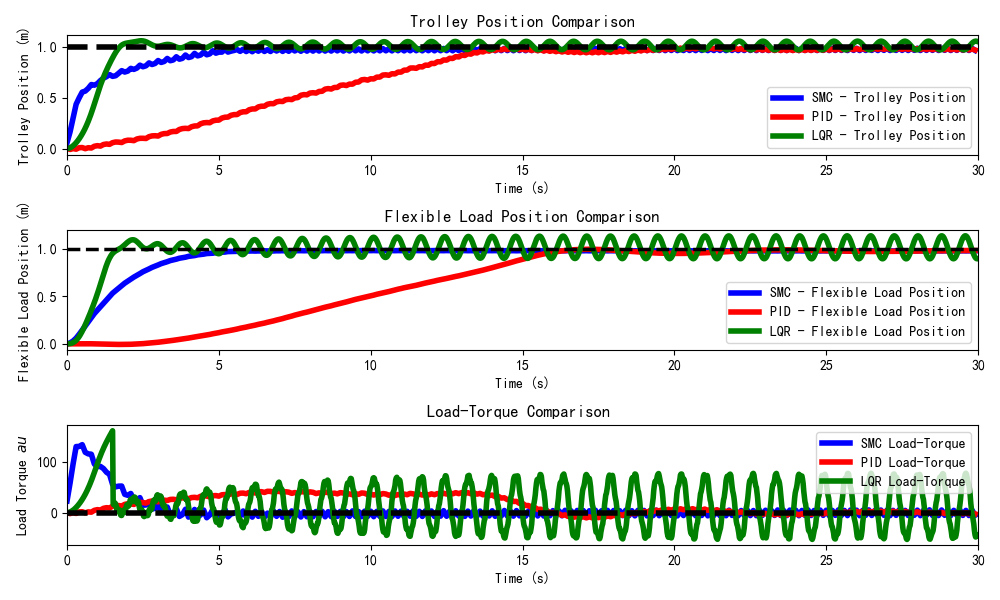}
    \caption{Comparison under low-frequency disturbance}
\end{subfigure}
\begin{subfigure}{0.5\textwidth}
    \includegraphics[width=\textwidth]{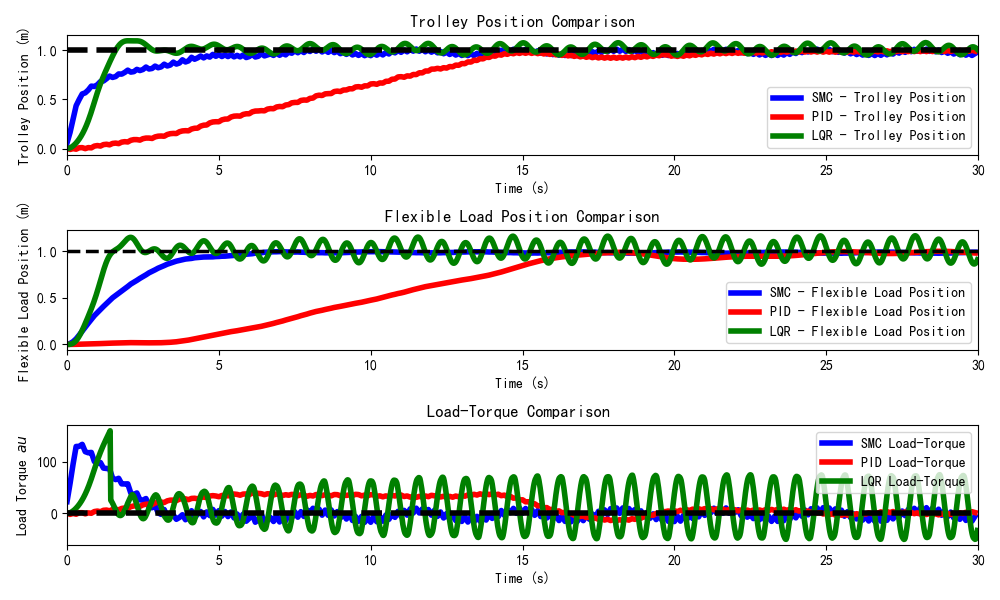}
    \caption{Comparison under high-frequency disturbance}
\end{subfigure}
\begin{subfigure}{0.5\textwidth}
    \includegraphics[width=\textwidth]{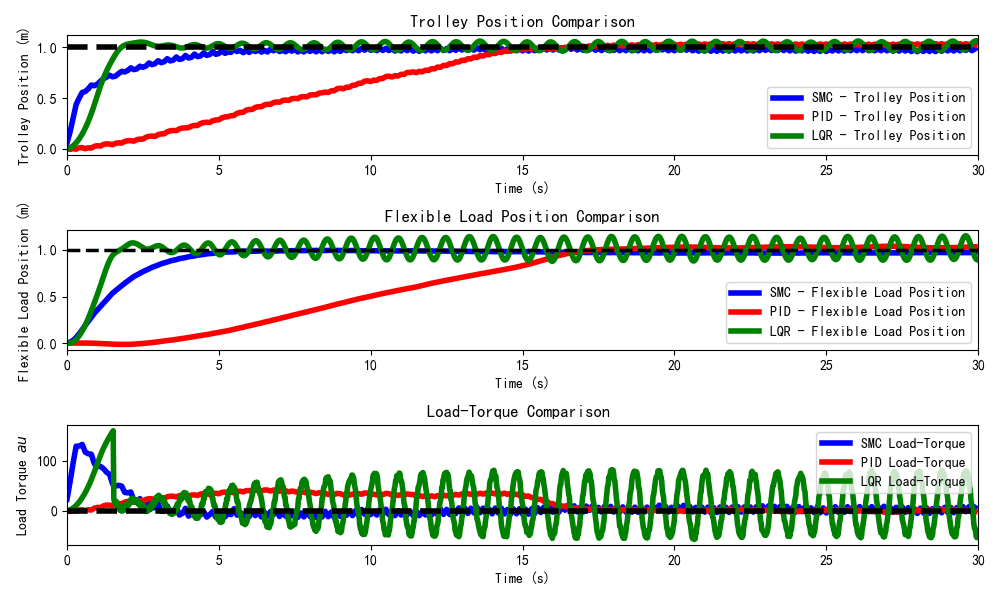}
    \caption{Comparison under random disturbance}
\end{subfigure}
\caption{Performance comparison of three controllers under different disturbance conditions}
\label{comparison_plot_lqr_pid}
\end{figure}

\begin{table}[h!]
\centering
\caption{Performance Comparison of Three Controllers}
\begin{tabular}{l c c c}
\toprule
\textbf{Controller} & \textbf{MSE} & \textbf{Max Error (m)} & \textbf{Response Time (s)} \\
\midrule
LQR                  & 0.11625     & 0.156               & 1.2                   \\
PID                  & 0.0153       & 0.05                 & 17.6                     \\
Proposed             &  0.0008      & 0.0556               & 7.3                     \\
\bottomrule
\end{tabular}
\label{table:comprehensive_comparison}
\end{table}

As shown in Table~\ref{table:comprehensive_comparison} and Fig.~\ref{comparison_plot_lqr_pid}, the proposed controller outperforms classical LQR and PID controllers in key performance metrics such as mean squared error (MSE), maximum error, and response time. It also demonstrates superior robustness and adaptability.

Specifically, the proposed controller significantly reduces the mean squared error (MSE), achieving only \textbf{0.69\%} of the LQR controller's MSE and \textbf{5.23\%} of the PID controller's MSE. In various disturbance scenarios, it exhibits precise target tracking, demonstrating strong trajectory-following capabilities.

Moreover, in terms of maximum error, the proposed controller also performs exceptionally well. Its maximum error is only \textbf{0.0556 m}, reducing the LQR controller's error by \textbf{64.4\%} and outperforming the PID controller's \textbf{0.05 m}. This indicates that the proposed controller effectively suppresses transient errors under high-frequency and random disturbances, ensuring system stability and accuracy.

Regarding response time, the proposed controller achieves a well-balanced performance. Its response time is \textbf{7.6 seconds}, slightly slower than the LQR controller's \textbf{1.2 seconds} but significantly better than the PID controller's \textbf{17.6 seconds}. Although the response speed is not as fast as the LQR controller, it avoids the severe oscillations observed in the LQR controller under high-frequency and random disturbances. The proposed controller smoothly transitions to the target position in a short time, demonstrating better dynamic adaptability.

From an overall evaluation perspective, the LQR controller requires precise system modeling and is prone to oscillations and steady-state errors under disturbances. The PID controller has strong disturbance rejection capabilities but suffers from slow response speed in dynamic environments. By introducing dynamic gain adjustment, adaptive mechanisms, and neural network compensation, the proposed controller not only improves trajectory tracking accuracy and stability but also exhibits strong robustness and adaptability under complex disturbance conditions. This makes the proposed controller an ideal choice for dynamic and complex environments, effectively balancing response speed, steady-state error, and disturbance rejection capabilities.

\section{Conclusion}
\label{sec:conclusion}

This paper proposes a disturbance-resistant controller for deep-sea cranes that integrates adaptive control, neural-network compensation, and hierarchical sliding-mode control to address nonlinear disturbances and uncertainties in the marine environment. By coordinating a global sliding surface with subsystem sliding surfaces, the controller guarantees global asymptotic stability, while the neural-network module enhances adaptability to complex disturbances. Simulation results confirm excellent robustness, high tracking accuracy, and fast response under various disturbance conditions, although the trolley expends more energy to maintain precision. Future work will therefore focus on optimizing the control strategy to lower energy consumption while further improving stability and adaptability in extreme marine scenarios.
\section{Acknowledgement}

This work was supported by the Natural Science Foundation of Hebei Province (F2024202028), and the Beijing-Tianjin-Hebei Basic Research Cooperation Special Project (F2024202118).

\bibliographystyle{unsrt}
\bibliography{references}

\end{document}